\documentclass{article}

\usepackage{fullpage,amsmath,amsthm,amssymb}
\usepackage{url}

\newtheorem{innercustomlem}{Restatement of Lemma}
\newenvironment{customlem}[1]
  {\renewcommand\theinnercustomlem{#1}\innercustomlem}
  {\endinnercustomlem}

\let\Pr\relax
\DeclareMathOperator*{\Pr}{\mathbb{P}}

\newcommand{\inprod}[1]{\left\langle #1 \right\rangle}
\newcommand{\eps}{\varepsilon}
\renewcommand{\log}{\lg}
\newcommand{\R}{\mathbb{R}}

\newcommand{\PointSets}{\mathcal{P}}

\newtheorem{theorem}{Theorem}
\newtheorem{lemma}{Lemma}
\newtheorem{conjecture}{Conjecture}
\newtheorem{corollary}{Corollary}

\title{Optimality of the Johnson-Lindenstrauss lemma}

\author{Kasper Green Larsen\thanks{Aarhus University. \texttt{larsen@cs.au.dk}. Supported by Center for Massive Data Algorithmics, a Center of the Danish National Research Foundation, grant DNRF84, a Villum Young Investigator Grant and an AUFF Starting Grant.}  \and Jelani Nelson\thanks{Harvard University. \texttt{minilek@seas.harvard.edu}. Supported by NSF CAREER award CCF-1350670, NSF grant IIS-1447471, ONR Young Investigator award N00014-15-1-2388, and a Google Faculty Research Award.}}

\begin{document}

\maketitle

\thispagestyle{empty}

\begin{abstract}
For any $d, n \geq 2$ and $1/({\min\{n,d\}})^{0.4999} <\eps<1$, we show the
existence of a set of $n$ vectors $X\subset \R^d$ such that any
embedding $f:X\rightarrow \R^m$ satisfying
$$
\forall x,y\in X,\ (1-\eps)\|x-y\|_2^2\le \|f(x)-f(y)\|_2^2 \le (1+\eps)\|x-y\|_2^2
$$
must have 
$$
m  = \Omega(\eps^{-2} \lg n).
$$
This lower bound matches the upper bound given by the Johnson-Lindenstrauss lemma \cite{JL84}. Furthermore, our lower bound holds for nearly the full range of $\eps$ of interest, since there is always an isometric embedding into dimension $\min\{d, n\}$ (either the identity map, or projection onto $\mathop{span}(X)$).

Previously such a lower bound was only known to hold against {\em linear} maps $f$, and not for such a wide range of parameters $\eps, n, d$ \cite{LarsenN16}. The best previously known lower bound for general $f$ was $m = \Omega(\eps^{-2}\lg n/\lg(1/\eps))$ \cite{Welch74,Levenshtein83,Alon03}, which is suboptimal for any $\eps = o(1)$.
\end{abstract}

\section{Introduction}

In modern algorithm design, often data is high-dimensional, and one seeks to first pre-process the data via some {\em dimensionality reduction} scheme that preserves geometry in such a way that is acceptable for particular applications. The lower-dimensional embedded data has the benefit of requiring less storage, less communication bandwith to be transmitted over a network, and less time to be analyzed by later algorithms. Such schemes have been applied to good effect in a diverse range of areas, such as streaming algorithms \cite{Muthukrishnan05}, numerical linear algebra \cite{Woodruff14}, compressed sensing \cite{CandesRT04,Donoho04}, graph sparsification \cite{SpielmanS11}, clustering \cite{BoutsidisZMD15,CohenEMMP15}, nearest neighbor search \cite{Har-PeledIM12}, and many others.

A cornerstone dimensionality reduction result is the following {\em Johnson-Lindenstrauss (JL) lemma} \cite{JL84}.
\begin{theorem}[JL lemma]
Let $X\subset \R^d$ be any set of size $n$, and let $\eps\in(0,1/2)$ be arbitrary. Then there exists a map $f:X\rightarrow \R^m$ for some $m = O(\eps^{-2}\log n)$ such that
\begin{equation}
\forall x,y\in X,\ (1-\eps)\|x-y\|_2^2\le \|f(x)-f(y)\|_2^2 \le (1+\eps)\|x-y\|_2^2 . \label{eqn:jl-guarantee}
\end{equation}
\end{theorem}
Even though the JL lemma has found applications in a plethora of different fields over the past three decades, its optimality has still not been settled. In the original paper by Johnson and Lindenstrauss~\cite{JL84}, it was proven that for any $\eps<1/2$, there exists $n$ point sets $X \subset \R^n$ for which any embedding $f : X \to \R^m$ providing \eqref{eqn:jl-guarantee} must have $m = \Omega(\lg n)$. This was later improved in \cite{Levenshtein83,Alon03}, which showed the existence of an $n$ point set $X \subset \R^n$, such that any $f$ providing \eqref{eqn:jl-guarantee} must have $m = \Omega(\min\{n, \eps^{-2} \lg n/\lg(1/\eps)\})$, which falls short of the JL lemma for any $\eps = o(1)$. This lower bound can also be obtained from the Welch bound \cite{Welch74}, which states $\eps^{2k} \ge (1/(n-1))(n/\binom{m+k-1}{k} - 1)$ for any positive integer $k$, by choosing $2k = \lceil \lg n / \lg(1/\eps)\rceil$. The lower bound can also be extended to hold for any $n \le e^{c\eps^2 d}$ for some constant $c>0$.

\paragraph{Our Contribution}
In this paper, we finally settle the optimality of the JL lemma. Furthermore, we do so for almost the full range of $\eps$.
\begin{theorem}
\label{thm:general}
For any integers $n,d \geq 2$ and $\eps \in (\lg^{0.5001}n/\sqrt{\min\{n,d\}}, 1)$, there exists a set of points $X\subset \R^d$ of size $n$, such that any map $f:X\rightarrow \R^m$ providing the guarantee \eqref{eqn:jl-guarantee} must have
\begin{equation}
m=\Omega(\eps^{-2}\lg (\eps^2 n)). \label{eqn:lb}
\end{equation}
\end{theorem}
Here it is worth mentioning that the JL lemma can be used to give an upper bound of
$$
m = O(\min\{n , d, \eps^{-2} \lg n\}),
$$
where the $d$ term is obvious (the identity map) and the $n$ term follows by projecting onto the $\leq n$-dimensional subspace spanned by $X$. Thus a requirement of at least $\eps = \Omega(1/\sqrt{\min\{n,d\}})$ is certainly necessary for the lower bound \eqref{eqn:lb} to be true, which our constraint on $\eps$ matches up to the $\lg^{0.5001} n$ factor. 

We also make the following conjecture concerning the behavior of the optimal form of Euclidean dimension reduction possible as $\eps\rightarrow 1/\sqrt{\min\{n,d\}}$. Note the $\log(\eps^2 n)$ term as opposed to $\log n$ in the upper bound.
\begin{conjecture}\label{conj:opt}
If $f(n,d,\eps)$ denotes the smallest $m$ such that all $n$-point subsets of $\ell_2^d$ can be embedded into $\ell_2^m$ with distortion at most $1+\eps$, then for all $n,d>1$ and $0<\eps<1$, $f(n,d,\eps) = \Theta(\min\{n, d, \eps^{-2}\log(2 + \eps^2 n)\})$.
\end{conjecture}

It is worth mentioning that the arguments in previous work \cite{Welch74,Alon03,LarsenN16} all produced hard point sets $P$ which were nearly orthogonal so that any embedding into an {\em incoherent} collection provided low distortion under the Euclidean metric. Recall $P$ is $\eps$-incoherent if every $x\in P$ has unit $\ell_2$ norm, and $\forall x\neq y \in P$ one has $|\langle x,y\rangle| = O(\eps)$. Unfortunately though, it is known that for any $\eps < 2^{-\omega(\sqrt{\lg n})}$, an incoherent collection of $n$ vectors in dimension $m = o(\eps^{-2}\lg n)$  exists, beating the guarantee of the JL lemma. The construction is based on Reed-Solomon codes (see for example \cite{AlonGHP92,NNW14}). Thus proving Theorem~\ref{thm:general} requires a very different construction of a hard point set when compared with previous work.

\subsection{Prior Work}
Prior to our work, a result of the authors \cite{LarsenN16} showed an $m = \Omega(\eps^{-2}\lg n)$ bound in the restricted setting where $f$ must be {\em linear}. This left open the possibility that the JL lemma could be improved upon by making use of nonlinear embeddings. Indeed, as mentioned above even the hard instance of \cite{LarsenN16} enjoys the existence of a nonlinear embedding into $m = o(\eps^{-2}\lg n)$ dimension for $\eps < 2^{-\omega(\sqrt{\lg n})}$. Furthermore, that result only provided hard instances with $n \le \mathop{poly}(d)$, and furthermore $n$ had to be sufficiently large (at least $\Omega(d^{1+\gamma}/\eps^2)$ for any constant $\gamma>0$).

Also related is the so-called {\em distributional JL} (DJL) lemma. The original proof of the JL lemma in \cite{JL84} is via {\em random projection}, i.e.\ ones picks a uniformly random rotation $U$ then defines $f(x)$ to be the projection of $Ux$ onto its first $m$ coordinates, scaled by $1/\sqrt{m}$ in order to have the correct squared Euclidean norm in expectation. Note that this construction of $f$ is both {\em linear}, and {\em oblivious} to the data set $X$. Indeed, all known proofs of the JL lemma proceed by instantiating distributions $\mathcal{D}_{\eps,\delta}$ satisfying the guarantee of the below distributional JL (DJL) lemma.

\begin{lemma}[Distributional JL (DJL) lemma]
For any integer $d\ge 1$ and any $0<\eps,\delta<1/2$, there exists a distribution $\mathcal{D}_{\eps,\delta}$ over $m\times d$ real matrices for some $m\lesssim \eps^{-2}\lg(1/\delta)$ such that
\begin{equation}
\forall u\in\R^d,\ \Pr_{\Pi\sim\mathcal{D}_{\eps,\delta}}(| \|\Pi u\|_2 - \|u\|_2 | > \eps\|u\|_2 ) < \delta . \label{eqn:djl}
\end{equation}
\end{lemma}

One then proves the JL lemma by proving the DJL lemma with $\delta < 1/\binom{n}{2}$, then performing a union bound over all $u\in \{x - y : x,y \in X\}$ to argue that $\Pi$ simultaneously preserves all norms of such difference vectors simultaneously with positive probability. It is known that the DJL lemma is tight \cite{JayramW13,KaneMN11}; namely any distribution $\mathcal{D}_{\eps,\delta}$ over $\R^{m\times n}$ satisfying \eqref{eqn:djl} must have $m = \Omega(\min\{d, \eps^{-2}\lg(1/\delta)\})$. Note though that, prior to our current work, it may have been possible to improve upon the JL lemma by avoiding the DJL lemma. Our main result implies that, unfortunately, this is not the case: obtaining \eqref{eqn:jl-guarantee} via the DJL lemma combined with a union bound is optimal.

\subsection{Subsequent Work}
After the initial dissemination of this work, Alon and Klartag asked the question of the optimal space complexity for solving the static ``approximate dot product'' problem on the sphere in $d$ dimensions \cite{AlonK17}. In this problem one is given a set $P$ of $n$ points $x_1,\ldots,x_n$ in $S^{d-1}$ to preprocess into a data structure, as well as an error parameter $\eps$. Then in response to \texttt{query}$(i,j)$, one must output $\langle x_i, x_j\rangle$ with additive error at most $\eps$. The work \cite{KushilevitzOR00} provides a solution using space $O(\eps^{-2}n\log n)$ bits, which turns out to be optimal iff $d = \Omega(\eps^{-2}\log n)$, shown by \cite{AlonK17}. In fact \cite{AlonK17} was able to provide an understanding of the precise asymptotic space complexity $s(n,d,\eps)$ of this problem for all ranges of $n,d,\eps$. This understanding as a consequence provides an alternate proof of the optimality of the JL lemma, since their work implies $s(n,n,2\eps) \gg s(n,c\eps^{-2}\log n,\eps)$ for $c>0$ a small constant (and if dimension-reduction into dimension $d'$ were always possible, one would have $s(n,n,2\eps) \le s(n,d',\eps)$ by first dimension-reducing the input!).

In terms of proof methods, unlike \cite{Alon03,Welch74}, our work uses an encoding argument. We proceed in a somewhat ad hoc fashion, showing that one can use simple upper bounds on the sizes of $\eps$-nets of various convex bodies to conclude that dimension reduction far below the JL upper bound would imply an encoding scheme that is too efficient to exist for some task, based on rounding vectors to net points (see Section~\ref{sec:proofoverview} for an overview). Interestingly enough, the original $m = \Omega(\log n)$ lower bound of \cite{JL84} was via a volumetric argument, which is related to the packing and covering bounds one needs to execute our encoding argument! The work of \cite{AlonK17} on understanding $s(n,d,\eps)$ is also via an encoding argument. They observe that the question of understanding $s(n,d,\eps)$ is essentially equivalent to understanding the logarithm of the optimal size of an $\eps$-net under entrywise $\ell_\infty$ norm of $n\times n$ Gram matrices of rank $d$, since $P$ can be encoded as the name of the closest point in the net to its Gram matrix. They then proceed to provide tight upper and lower bounds on the optimal net size for the full range of parameters.

The work \cite{AlonK17} also made progress toward Conjecture~\ref{conj:opt}. In particular, they proved the lower bound for all ranges of parameters, thus removing the ``$\log^{0.5001} n$'' term in our requirement on $\eps$ in Theorem~\ref{thm:general}. As for the upper bound, they made progress on a {\em bipartite} version of the conjecture. In particular, they showed that for any $2n$ vectors $x_1,\ldots,x_n,y_1,\ldots,y_n\in S^{d-1}$, one can find $2n$ vectors $a_1,\ldots,a_n,b_1,\ldots,b_n\in S^{m-1}$ for $m = O(\eps^{-2}\log(2 + \eps^2 n))$ so that for all $i,j\in[n]$, $|\langle x_i,y_j\rangle - \langle a_i, b_j\rangle| < \eps$. No promise is given for dot product preservation amongst the $x_i$'s internally, or amongst the $y_j$'s internally. Also note that dot product preservation up to additive $\eps$ error does not always imply norm preservation with relative error $1+\eps$, i.e.\ when distances are small. 

\section{Preliminaries on Covering Convex Bodies}
We here state a standard result on covering numbers. The proof is via a volume comparison argument; see for example \cite[Equation (5.7)]{Pisier89}.

\begin{lemma}
\label{lem:coveringL2}
Let $E$ be an $m$-dimensional normed space, and let $B_E$ denote its unit ball. For any $0 < \eps < 1$, one can cover $B_E$ using at most $2^{m \lg(1 + 2/\eps)}$ translated copies of $\eps B_E$.
\end{lemma}

\begin{corollary}
\label{cor:coveringC}
Let $T$ be an origin symmetric convex body in $\R^m$. For any $0 < \eps < 1$, one can cover $T$ using at most $2^{m\lg(1 + 2/\eps)}$ translated copies of $\eps T$.
\end{corollary}

\begin{proof}
The Minkowski functional of an origin symmetric convex body $T$, when restricted to the subspace spanned by vectors in $T$, is a norm for which $T$ is the unit ball (see e.g.~\cite[Proposition 1.1.8]{Thompson96}). It thus follows from Lemma~\ref{lem:coveringL2} that $T$ can be covered using at most $2^{m\lg(1 + 2/\eps)}$ translated copies of $\eps T$.
\end{proof}

In the remainder of the paper, we often use the notation $B_p^d$ to denote the unit $\ell_p$ ball in $\R^d$.

\section{Lower Bound Proof}
\label{sec:proofoverview}
In the following, we start by describing the overall strategy in our proof. This first gives a fairly simple proof of a sub-optimal lower bound. We then introduce the remaining ideas needed and complete the full proof. The proof goes via a counting argument. More specifically, we construct a large family $\PointSets = \{P_1,P_2,\dots\}$ of very different sets of $n$ points in $\R^d$. We then assume all point sets in $\PointSets$ can be embedded into $\R^m$ while preserving all pairwise distances to within $(1+\eps)$. Letting $f_1(P_1),f_2(P_2),\dots,$ denote the embedded point sets, we then argue that our choice of $\PointSets$ ensures that any two $f_i(P_i)$ and $f_j(P_j)$ must be very different. If $m$ is too low, this is impossible as there are not enough sufficiently different point sets in $\R^m$.

In greater detail, the point sets in $\PointSets$ are chosen as follows: Let $e_1,\dots,e_d$ denote the standard unit vectors in $\R^d$. For now, assume that $d = n/\lg(1/\eps)$ and $\eps \in (\lg^{0.5001}n/\sqrt{d}, 1)$. We will later show how to generalize the proof to the full range of $d$. For any set $S \subset [d]$ of $k=\eps^{-2}/256$ indices, define a vector $y_S := \sum_{j \in S} e_j/\sqrt{k}$. A vector $y_S$ has the property that $\langle y_S, e_j \rangle = 0$ if $j \notin S$ and $\langle y_S, e_j \rangle = 16 \eps$ if $j \in S$. The crucial property here is that there is a gap of $16 \eps$ between the inner products depending on whether or not $j \in S$. Now if $f$ is a mapping to $\R^m$ that satisfies the JL-property~\eqref{eqn:jl-guarantee} for $P=\{0,e_1,\dots,e_d,y_S\}$, then first off, we can assume $f(0)=0$ since pairwise distances are translation invariant. From this it follows that $f$ must preserve norms of the vectors $x \in P$ to within $(1+\eps)$ since 
\begin{align*}
(1-\eps)\|x\|_2^2 &= (1-\eps)\|x-0\|_2^2 \leq \|f(x)-f(0)\|_2^2 \\
{}&= \|f(x)\|_2^2 = \|f(x)-f(0)\|_2^2 \\
{}& \leq (1+\eps)\|x-0\|_2^2\\
{}& = (1+\eps) \|x\|_2^2.
\end{align*}
We then have that $f$ must preserve inner products $\langle e_j, y_S \rangle$ up to an additive of $4 \eps$. This can be seen by the following calculations, where $v \pm X$ denotes the interval $[v-X, v+X]$:
\begin{eqnarray*}
\|f(e_j)-f(y_S)\|_2^2 &=& \|f(e_j)\|_2^2 + \|f(y_S)\|_2^2\\
&&{}- 2\langle f(e_j), f(y_S)\rangle \Rightarrow \\
2\langle f(e_j), f(y_S)\rangle &\in& (1\pm \eps)\|e_j\|_2^2 + (1\pm \eps)\|y_S\|_2^2\\
&&{}- (1\pm \eps)\|e_j-y_S\|_2^2 \Rightarrow \\
2\langle f(e_j), f(y_S)\rangle &\in& 2\langle e_j, y_S\rangle \pm \eps(\|e_j\|_2^2 + \|y_S\|_2^2\\
&&{} + \|e_j-y_S\|_2^2) \Rightarrow \\
\langle f(e_j), f(y_S)\rangle &\in& \langle e_j, y_S\rangle \pm 4\eps.
\end{eqnarray*}
This means that after applying $f$, there remains a gap of $(16-8)\eps = 8\eps$ between $\langle f(e_j),f(y_S)\rangle$ depending on whether or not $j \in S$. With this observation, we are ready to describe the point sets in $\PointSets$ (in fact they will not be point sets, but rather ordered sequences of points, possibly with repetition). Let $Q = n-d-1$. For every choice of $Q$ sets $S_1,\dots,S_Q \subset [d]$ of $k$ indices each, we add a point set $P$ to $\PointSets$. The sequence $P$ is simply $(0,e_1,\dots,e_d,y_{S_1},\dots,y_{S_Q})$. This gives us a family $\PointSets$ of size $\binom{d}{k}^Q$. If we look at JL embeddings for all of these point sets $f_1(P_1),f_2(P_2),\dots$, then intuitively these embeddings have to be quite different. This is true since $f_i(P_i)$ uniquely determines $P_i$ simply by computing all inner products between the $f_i(e_j)$'s and $f_i(y_{S_\ell})$'s. The problem we now face is that there are infinitely many sets of $n$ points in $\R^m$ that one can embed to. We thus need to discretize $\R^m$ in a careful manner and argue that there are not enough $n$-sized sets of points in this discretization to uniquely embed each $P_i$ when $m$ is too low.

\paragraph{Encoding Argument} To give a formal proof that there are not enough ways to embed the point sets in $\PointSets$ into $\R^m$ when $m$ is low, we give an encoding argument. More specifically, we assume that it is possible to embed every point set in $\PointSets$ into $\R^m$ while preserving pairwise distances to within $(1+\eps)$. We then present an algorithm that based on this assumption can take any point set $P \in \PointSets$ and encode it into a bit string of length $O(nm)$. The encoding guarantees that $P$ can be uniquely recovered from the encoding. The encoding algorithm thus effectively defines an injective mapping $g$ from $\PointSets$ to $\{0,1\}^{O(nm)}$. Since $g$ is injective, we must have $|\PointSets| \leq 2^{O(nm)}$. But $|\PointSets| = \binom{d}{k}^Q = (\eps^2 n/\lg(1/\eps))^{\Omega(\eps^{-2}n)}$ and we can conclude $m = \Omega(\eps^{-2}\lg(\eps^2 n/\lg(1/\eps)))$. For $\eps>1/n^{0.4999}$, this is $m = \Omega(\eps^{-2}\lg n)$.

\paragraph{First Attempt}
The difficult part is to design an encoding algorithm that yields an encoding of size $O(nm)$ bits. A natural first attempt would go as follows: recall that any JL-embedding $f$ for a point set $P \in \PointSets$ (where $f$ may depend on $P$) must preserve gaps in $\langle f(e_j),f(y_{S_\ell})\rangle$'s depending on whether or not $j \in S_\ell$. This follows simply by preserving distances to within a factor $(1+\eps)$ as argued above. If we can give an encoding that allows us to recover approximations $\hat{f}(e_j)$ of $f(e_j)$ and $\hat{f}(y_{S_\ell})$ of $f(y_{S_\ell})$ such that $\|\hat{f}(e_j) - f(e_j)\|_2^2 \leq \eps$ and $\|\hat{f}(y_{S_\ell}) - f(y_{S_\ell})\|_2^2 \leq \eps$, then by the triangle inequality, the distance $\|\hat{f}(e_j) -\hat{f}(y_{S_\ell})\|_2^2$ is also a $(1+O(\eps))$ approximation to $\|e_j - y_{S_\ell}\|_2^2$ and the gap between inner products would be preserved. To encode sufficiently good approximations $\hat{f}(e_j)$ and $\hat{f}(y_{S_\ell})$, one could do as follows: since norms are roughly preserved by $f$, we must have $\|f(e_j)\|_2^2, \|f(y_{S_\ell})\|_2^2 \leq 1+\eps$. Letting $B_2^m$ denote the $\ell_2$ unit ball in $\R^m$, we could choose some fixed covering $C_2$ of $(1+\eps)B_2^m$ with translated copies of $\eps B_2^m$. Since $f(e_j), f(y_{S_\ell}) \in (1+\eps)B_2^m$, we can find translations $c_2(f(e_j)) + \eps B_2^m$ and $c_2(f(y_{S_\ell})) + \eps B_2^m$ of $\eps B_2^m$ in $C_2$, such that these balls contain $f(e_j)$ and $f(y_{S_\ell})$ respectively. Letting $\hat{f}(e_j) = c_2(f(e_j)) $ and $\hat{f}(y_{S_\ell})=c_2(f(y_{S_\ell}))$ be the centers of these balls, we can encode an approximation of $f(e_j)$ and $f(y_{S_\ell})$ using $\lg |C_2|$ bits by specifying indices into $C_2$. Unfortunately, covering $(1+\eps)B_2^m$ by $\eps B_2^m$ needs $|C_2| = 2^{\Omega(m\lg(1/\eps))}$ since the volume ratio between $(1+\eps)B_2^m$ and $\eps B_2^m$ is $(1/\eps)^{\Omega(m)}$. The $\lg(1/\eps)$ factor loss leaves us with a lower bound on $m$ of no more than $m = \Omega(\eps^{-2}\lg(\eps^2 n/\lg(1/\eps))/\lg(1/\eps))$, roughly recovering the lower bound of Alon \cite{Alon03} by a different argument.

\paragraph{Full Proof}
The key idea to reduce the length of the encoding to $O(nm)$ is as follows: First observe that we chose $d=n/\lg(1/\eps)$. Thus we can spend up to $O(m \lg(1/\eps))$ bits encoding each $f(e_j)$'s. Thus we simply encode approximations $\hat{f}(e_j)$ by specifying indices into a covering $C_2$ of $(1+\eps)B_2^m$ by $\eps B_2^m$ as outlined above. 

For the $f(y_{S_\ell})$'s, we have to be more careful as we cannot afford $m \lg(1/\eps)$ bits for each. First, we define the $d \times m$ matrix $A$ having the $\hat{f}(e_j) = c_2(f(e_j))$ as rows (see Figure~\ref{fig:main-idea}). Note that this matrix can be reconstructed from the part of the encoding specifying the $\hat{f}(e_j)$s. Now observe that the $j$'th coordinate of $v_\ell=A f(y_{S_\ell})$ is equal to $\langle \hat{f}(e_j),  f(y_{S_\ell}) \rangle$. This is within $O(\eps)$ of $\langle e_j, y_{S_\ell} \rangle$. The coordinates of $v_\ell$ thus determine $S_\ell$ due to the gap in inner products depending on whether $j \in S_\ell$ or not. We therefore seek to encode the $v_\ell$ efficiently. Since the $v_\ell$ are in $\R^d$, this seems quite hopeless to do in $O(m)$ bits per $v_\ell$. The key observation is that they lie in an $m$-dimensional subspace of $\R^d$, namely in the column space of $A$. This observation will allow us to get down to just $O(m)$ bits. We are ready to give the remaining details.

\begin{figure}
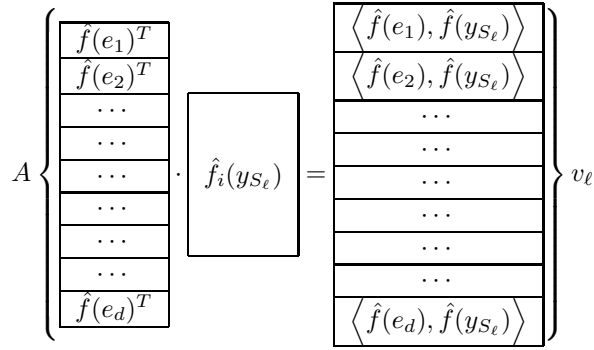

\begin{center}
$$
A\left\{\begin{tabular}{|c|}
\hline
$\hat{f}(e_1)^T$\\
\hline
$\hat{f}(e_2)^T$\\
\hline
$\cdots$\\
\hline
$\cdots$\\
\hline
$\cdots$\\
\hline
$\cdots$\\
\hline
$\cdots$\\
\hline
$\cdots$\\
\hline
$\hat{f}(e_d)^T$\\
\hline
\end{tabular}\cdot
\begin{tabular}{|c|}
\hline\\
\\
$\hat{f}_i(y_{S_\ell})$\\
\\
\\
\hline
\end{tabular}
= 
\begin{tabular}{|c|}
\hline
$\inprod{\hat{f}(e_1), \hat{f}(y_{S_\ell})}$\\
\hline
$\inprod{\hat{f}(e_2), \hat{f}(y_{S_\ell})}$\\
\hline
$\cdots$\\
\hline
$\cdots$\\
\hline
$\cdots$\\
\hline
$\cdots$\\
\hline
$\cdots$\\
\hline
$\cdots$\\
\hline
$\inprod{\hat{f}(e_d), \hat{f}(y_{S_\ell})}$\\
\hline
\end{tabular}\right\}v_\ell
$$
\caption{Notation to describe a more efficient encoding of $P\in\mathcal P$.}\label{fig:main-idea}
\end{center}
\end{figure}

Let
$W$ denote the subspace of $\R^d$ spanned by the columns of
$A$. We have $\dim(W) \leq m$. Define $T$ as the convex body 
$$
T:= B_\infty^d \cap W.
$$
That is, $T$ is the intersection of the subspace
$W$ with the $d$-dimensional $\ell_\infty$ unit ball
$B_\infty^d$. Now let $C_\infty$ be a minimum cardinality
covering of $(22 \eps)T$ by
translated copies of $\eps T$, computed by any deterministic procedure that
depends only on $T$. Since $T$ is origin symmetric, by Corollary~\ref{cor:coveringC} it follows that
$|C_\infty| \leq 2^{m \lg 45}$. To encode the vectors
$y_{S_1},\dots,y_{S_Q}$ we make use of the following lemma, whose proof we give in Section~\ref{sec:closeproof}:
\begin{lemma}
\label{lem:closeIPs}
For every $e_j$ and $y_{S_\ell}$ in $P$, we have 
$$|\langle \hat{f}(e_j), f(y_{S_\ell}) \rangle - \langle e_j ,y_{S_\ell} \rangle| \leq 6\eps.$$
\end{lemma}
From Lemma~\ref{lem:closeIPs}, it
follows that $|\langle \hat{f}(e_j), f(y_{S_\ell}) \rangle| \leq 6\eps +
\langle e_j, y_{S_\ell} \rangle \leq 22\eps$ for every
$e_j$ and $y_{S_\ell}$ in $P$. Since the $j$'th coordinate of $Af(y_{S_\ell})$
equals $\langle \hat{f}(e_j), f(y_{S_\ell}) \rangle$, it follows that $Af(y_{S_\ell})
\in (22 \eps)T$. Using this fact, we encode each $y_{S_\ell}$ by finding some
vector $c_\infty(y_{S_\ell})$ such that
$c_\infty(y_{S_\ell})+\eps T$ is a convex shape in
the covering $C_\infty$ and $Af(y_{S_\ell}) \in c_\infty(y_{S_\ell})
+\eps T$. We write down $c_\infty(y_{S_\ell})$ as
an index into $C_\infty$. This costs a total of $Qm\lg 45 = O(Qm)$ bits over all
$y_{S_\ell}$. We now describe our decoding algorithm.

\paragraph{Decoding Algorithm}
To recover $P = \{0,e_1,\dots,e_d,y_{S_1},\dots,y_{S_Q}\}$ from the above encoding, we only have to recover
$y_{S_1},\dots,y_{S_Q}$ as $\{0,e_1,\dots,e_d\}$ is the same for all $P
\in \PointSets$. We first reconstruct the matrix $A$. We can do this
since $C_2$ was chosen independently of $P$ and thus by the indices
encoded into $C_2$, we recover $c_2(e_j) = \hat{f}(e_j)$ for
$j=1,\dots,d$. These are the rows of $A$. Then given $A$, we know $T$. Knowing $T$, we compute $C_\infty$ since it was
constructed via a deterministic procedure depending only on $T$. This
finally allows us to recover $c_\infty(y_{S_1}),\dots,c_\infty(y_{S_Q})$. What
remains is to recover $y_{S_1},\dots,y_{S_Q}$. Since $y_{S_\ell}$ is uniquely
determined from the set $S_\ell \subseteq \{1,\dots,d\}$ of
$k$ indices, we focus on recovering this set of indices
for each $y_{S_\ell}$.

For $\ell=1,\dots,Q$ recall that $Af(y_{S_\ell})$ is in $c_\infty(y_{S_\ell})+\eps T$. Observe now that:
\begin{eqnarray*}
Af(y_{S_\ell}) \in  c_\infty(y_{S_\ell}) + \eps T &\Rightarrow& \\
Af(y_{S_\ell}) - c_\infty(y_{S_\ell}) \in \eps T &\Rightarrow& \\
\|Af(y_{S_\ell})-c_\infty(y_{S_\ell})\|_\infty \leq \eps.
\end{eqnarray*}
But the $j$'th coordinate of $Af(y_{S_\ell})$ is $\langle \hat{f}(e_j), f(y_{S_\ell})
\rangle$. We combine the above with Lemma~\ref{lem:closeIPs} to deduce
$|(c_\infty(y_{S_\ell}))_j - \langle e_j , y_{S_\ell} \rangle| \leq 7\eps$ for all $j$. We thus have that $(c_\infty(y_{S_\ell}))_j \leq
7 \eps$ for $j \notin S_i$ and $(c_\infty(y_{S_\ell}))_j \geq 9 \eps$ for $j
\in S_\ell$. We finally conclude that the set $S_\ell$, and thus $y_{S_\ell}$, is
uniquely determined from $c_\infty(y_{S_\ell})$.

\paragraph{Analysis}
We finally analyse the size of the encoding produced by the above
procedure and derive a lower bound on $m$. Recall that the encoding
procedure produces a total of $d m \lg(1+4/\eps) + O(Q m)
 = O(nm)$
bits. But $|\PointSets|
\geq \left(\binom{d}{k}/2\right)^{Q} \geq (d/(2k))^{kQ} =
(d/(2k))^{k(n-d-1)} \geq (d/(2k))^{kn/2}$. We therefore must have
\begin{eqnarray*}
nm &=& \Omega(kn \lg(d/k))
\Rightarrow\\
m &=& \Omega(\eps^{-2} \lg( \eps^2 n /\lg(1/\eps))).
\end{eqnarray*}
Since we assume $\eps > \lg^{0.5001}n/\sqrt{d} \geq \lg^{0.5001}n/\sqrt{n}$, this can be
simplified to
$$
m = \Omega(\eps^{-2} \lg(\eps^2 n)).
$$
This shows that $m=\Omega(\eps^{-2} \lg(\eps^2 n))$ for $d = n/\lg(1/\eps)$ and $\eps \in (\lg^{0.5001}n/\sqrt{d}, 1)$. The following paragraph shows how to handle the remaining values of $d$.

\paragraph{Handling Other Values of $d$}
For $d > n/\lg(1/\eps)$, the proof is easy: Simply repeat the above construction using only the first $n/\lg(1/\eps)$ standard unit vectors in the point sets of $\PointSets$. This reproves the above lower bound, with the only further restriction that $\eps \in (\lg^{0.5001}n/\sqrt{\min\{d,n\}}, 1)$ as opposed to $\eps \in (\lg^{0.5001}n/\sqrt{d}, 1)$. 

For $d < n/\lg(1/\eps)$ and $\eps \in (\lg^{0.5001}n/\sqrt{d}, 1)$, assume for the sake of contradiction that it is possible to embed into $o(\eps^{-2} \lg(\eps^2 n))$ dimensions. Now take any point set $P$ in $\R^{d'}$ with $d' = n/\lg(1/\eps)$ and apply a JL transform into $d$ dimensions on it, obtaining a point set $P'$ in $d$ dimensions. This new point set has all distances preserved to within $(1 + O(\sqrt{\lg n/ d}))$ (by the standard JL upper bound). Next apply the hypothetical JL transform in $d$ dimensions to reduce the target dimension to $o(\eps^{-2} \lg(\eps^2 n))$. Distances are now preserved to within $(1 + O(\sqrt{\lg n/ d}))(1+\eps)$. Since we assumed $\eps > \lg^{0.5001} n/\sqrt{d}$, we have that $(1+O(\sqrt{\lg n/d})) = (1+o(\eps))$, which implies $(1 + O(\sqrt{\lg n/ d}))(1+\eps) =(1+O(\eps))$. This contradicts the lower bound for $d' = n/\lg(1/\eps)$ dimensions.

\subsection{Proof of Lemma~\ref{lem:closeIPs}}\label{sec:closeproof}
\label{sec:proofLemma}
In this section, we prove the lemma:
\begin{customlem}{\ref{lem:closeIPs}}
For every $e_j$ and $y_{S_\ell}$ in $P$, we have 
$$|\langle \hat{f}(e_j), f(y_{S_\ell}) \rangle - \langle e_j ,y_{S_\ell} \rangle| \leq 6\eps.$$
\end{customlem}
\begin{proof}
First note that:
\begin{eqnarray*}
\langle \hat{f}(e_j), f(y_{S_\ell}) \rangle &=&  \\ \langle c_2(e_j)-f(e_j) + f(e_j)
,f(y_{S_\ell}) \rangle 
&=& \\ \langle f(e_j), f(y_{S_\ell})
\rangle + \langle c_2(e_j)-f(e_j), f(y_{S_\ell}) \rangle
&\in& \\ \langle f(e_j), f(y_{S_\ell})
\rangle \pm \|c_2(e_j)-f(e_j)\|_2 \| f(y_{S_\ell})\|_2.
\end{eqnarray*}
Since $C_2$ was a covering with $\eps B_2^m$, we have
$\|c_2(e_j)-f(e_j)\|_2 \leq \eps$. Recall that $\|f(y_{S_\ell})\|^2_2 \leq
(1+\eps)$. This in particular implies that $\|f(y_{S_\ell})\|_2 \leq 2$. We thus have:
\begin{eqnarray}
\label{eq:range}
\langle \hat{f}(e_j), f(y_{S_\ell})) \rangle &\in&  \langle f(e_j), f(y_{S_\ell})
\rangle \pm 2\eps.
\end{eqnarray}
To bound $\langle f(e_j), f(y_{S_\ell}) \rangle$, observe that
\begin{eqnarray*}
\|f(e_j)-f(y_{S_\ell})\|_2^2 &=&  \\ \|f(e_j)\|_2^2 + \|f(y_{S_\ell})\|_2^2 - 2\langle
f(e_j), f(y_{S_\ell}) \rangle.
\end{eqnarray*}
This implies that
\begin{eqnarray*}
2\langle f(e_j),f(y_{S_\ell}) \rangle &\in& \\
 \|e_j\|_2^2(1\pm \eps) +
\|y_{S_\ell}\|_2^2(1\pm \eps) - \|e_i-y_{S_\ell}\|_2^2(1 \pm \eps) 
&\subseteq& \\ 2\langle e_j ,y_{S_\ell} \rangle \pm \eps(\|e_j\|_2^2 + \|y_{S_\ell} \|_2^2 +
\|e_j-y_{S_\ell}\|_2^2) 
&\subseteq& \\ 2\langle e_j ,y_{S_\ell} \rangle \pm \eps(4(\|e_j\|_2^2 + \|y_{S_\ell} \|_2^2)) 
\end{eqnarray*}
That is, we have
\begin{eqnarray*}
\langle f(e_j),f(y_{S_\ell})\rangle &\in& \langle e_j ,y_{S_\ell} \rangle \pm 2 \eps 
(\|e_j\|_2^2 + \|y_{S_\ell}\|_2^2) 
\end{eqnarray*}
Both the $e_j$'s and $y_{S_\ell}$'s have unit norm, hence
\begin{eqnarray*}
\langle f(e_j),f(y_{S_\ell})\rangle &\in& \langle e_j ,y_{S_\ell} \rangle \pm 4\eps
\end{eqnarray*}
Inserting this in~\eqref{eq:range}, we obtain
\begin{eqnarray*}
\langle \hat{f}(e_j), f(y_{S_\ell}) \rangle &\in& \langle e_j ,y_{S_\ell} \rangle \pm 6\eps.
\end{eqnarray*}
\end{proof}

\section*{Acknowledgments}
We thank Oded Regev for pointing out a simplification to our initial argument for handling the case $d < n/\lg(1/\eps)$, and for his permission to include the simpler argument here.

K.G.L. is supported by Center for Massive Data Algorithmics, a Center of the Danish National Research Foundation, grant DNRF84, a Villum Young Investigator Grant and an AUFF Starting Grant. J.N. did this work while supported by NSF grant IIS-1447471 and CAREER award CCF-1350670, ONR Young Investigator award N00014-15-1-2388, and a Google Faculty Research Award.

\bibliographystyle{alpha}

\begin{thebibliography}{CEM{\etalchar{+}}15}

\bibitem[AGHP92]{AlonGHP92}
Noga Alon, Oded Goldreich, Johan H{\aa}stad, and Ren\'{e} Peralta.
\newblock Simple construction of almost k-wise independent random variables.
\newblock {\em Random Struct. Algorithms}, 3(3):289--304, 1992.

\bibitem[AK17]{AlonK17}
Noga Alon and Bo'az Klartag.
\newblock Optimal compression of approximate inner products and dimension
  reduction.
\newblock In {\em Proceedings of the 58th Annual Symposium on Foundations of
  Computer Science (FOCS)}, 2017.

\bibitem[Alo03]{Alon03}
Noga Alon.
\newblock Problems and results in extremal combinatorics--{I}.
\newblock {\em Discrete Mathematics}, 273(1-3):31--53, 2003.

\bibitem[BZMD15]{BoutsidisZMD15}
Christos Boutsidis, Anastasios Zouzias, Michael~W. Mahoney, and Petros Drineas.
\newblock Randomized dimensionality reduction for k-means clustering.
\newblock {\em {IEEE} Transactions on Information Theory}, 61(2):1045--1062,
  2015.

\bibitem[CEM{\etalchar{+}}15]{CohenEMMP15}
Michael~B. Cohen, Sam Elder, Cameron Musco, Christopher Musco, and
  M\u{a}d\u{a}lina Persu.
\newblock Dimensionality reduction for k-means clustering and low rank
  approximation.
\newblock In {\em Proceedings of the 47th ACM Symposium on Theory of Computing
  (STOC)}, 2015.
\newblock Full version at \url{http://arxiv.org/abs/1410.6801v3}.

\bibitem[CRT06]{CandesRT04}
Emmanuel Cand\`{e}s, Justin Romberg, and Terence Tao.
\newblock Robust uncertainty principles: exact signal reconstruction from
  highly incomplete frequency information.
\newblock {\em IEEE Trans. Inf. Theory}, 52(2):489--509, 2006.

\bibitem[Don06]{Donoho04}
David Donoho.
\newblock Compressed sensing.
\newblock {\em IEEE Trans. Inf. Theory}, 52(4):1289--1306, 2006.

\bibitem[HIM12]{Har-PeledIM12}
Sariel Har{-}Peled, Piotr Indyk, and Rajeev Motwani.
\newblock Approximate nearest neighbor: Towards removing the curse of
  dimensionality.
\newblock {\em Theory of Computing}, 8(1):321--350, 2012.

\bibitem[JL84]{JL84}
William~B. Johnson and Joram Lindenstrauss.
\newblock Extensions of {Lipschitz} mappings into a {Hilbert} space.
\newblock {\em Contemporary Mathematics}, 26:189--206, 1984.

\bibitem[JW13]{JayramW13}
T.~S. Jayram and David~P. Woodruff.
\newblock Optimal bounds for {Johnson}-{Lindenstrauss} transforms and streaming
  problems with subconstant error.
\newblock {\em {ACM} Transactions on Algorithms}, 9(3):26, 2013.

\bibitem[KMN11]{KaneMN11}
Daniel~M. Kane, Raghu Meka, and Jelani Nelson.
\newblock Almost optimal explicit {Johnson}-{Lindenstrauss} families.
\newblock In {\em Proceedings of the 15th International Workshop on
  Randomization and Computation (RANDOM)}, pages 628--639, 2011.

\bibitem[KOR00]{KushilevitzOR00}
Eyal Kushilevitz, Rafail Ostrovsky, and Yuval Rabani.
\newblock Efficient search for approximate nearest neighbor in high dimensional
  spaces.
\newblock {\em {SIAM} J. Comput.}, 30(2):457--474, 2000.

\bibitem[LN16]{LarsenN16}
Kasper~Green Larsen and Jelani Nelson.
\newblock The {Johnson}-{Lindenstrauss} lemma is optimal for linear
  dimensionality reduction.
\newblock In {\em Proceedings of the 43rd International Colloquium on Automata,
  Languages and Programming (ICALP)}, 2016.

\bibitem[Lev83]{Levenshtein83}
Vladimir~I.~Levenshtein.
\newblock Bounds for packings of metric spaces and some of their applications (in Russian).
\newblock {\em Problemy Tekhn. Kibernet. Robot.} vol.\ 40, pages 43--110, 1983.

\bibitem[Mut05]{Muthukrishnan05}
S.~Muthukrishnan.
\newblock Data streams: Algorithms and applications.
\newblock {\em Foundations and Trends in Theoretical Computer Science}, 1(2),
  2005.

\bibitem[NNW14]{NNW14}
Jelani Nelson, Huy~L. Nguy$\tilde{\hat{\mbox{e}}}$n, and David~P. Woodruff.
\newblock On deterministic sketching and streaming for sparse recovery and norm
  estimation.
\newblock {\em Linear Algebra and its Applications, Special Issue on Sparse
  Approximate Solution of Linear Systems}, 441:152--167, 2014.

\bibitem[Pis89]{Pisier89}
Gilles Pisier.
\newblock {\em The volume of convex bodies and {Banach} space geometry},
  volume~94 of {\em Cambridge Tracts in Mathematics}.
\newblock Cambridge University Press, 1989.

\bibitem[SS11]{SpielmanS11}
Daniel~A. Spielman and Nikhil Srivastava.
\newblock Graph sparsification by effective resistances.
\newblock {\em {SIAM} J. Comput.}, 40(6):1913--1926, 2011.

\bibitem[Tho96]{Thompson96}
Anthony~C. Thompson.
\newblock {\em Minkowski Geometry}.
\newblock Encyclopedia of Mathematics and Its Applications. Cambridge
  University Press, 1996.

\bibitem[Wel74]{Welch74}
Lloyd~R. Welch.
\newblock Lower bounds on the maximum cross correlation of signals.
\newblock {\em IEEE Transactions on Information Theory}, 20, May 1974.

\bibitem[Woo14]{Woodruff14}
David~P. Woodruff.
\newblock Sketching as a tool for numerical linear algebra.
\newblock {\em Foundations and Trends in Theoretical Computer Science},
  10(1-2):1--157, 2014.

\end{thebibliography}

\newcommand{\etalchar}[1]{$^{#1}$}

\end{document}